\documentclass[review]{elsarticle}

\usepackage{hyperref}
\usepackage{amsthm}
\usepackage{amssymb}
\usepackage{mathtools}

\newtheorem{Remark}{Remark}
\newtheorem{Definition}{Definition}
\newtheorem{Theorem}{Theorem}
\newtheorem{Example}{Example}
\newtheorem{Lemma}[Theorem]{Lemma}

\journal{arXiv}

\begin{document}

\begin{frontmatter}


\title{Construction of Lyapunov Functions Using Vector Field Decomposition}

\author[mymainaddress]{Yuanyuan Liu\corref{mycorrespondingauthor}}
\cortext[mycorrespondingauthor]{Corresponding author}
\ead{liuyuanyuan@mail.nwpu.edu.cn}

\address[mymainaddress]{Northwestern Polytechnical University, Xi’an 710072, China.}


%
%

\begin{abstract}
In the present paper, a novel vector field decomposition based approach for constructing Lyapunov functions is proposed. For a given dynamical system, if the defining vector field admits a decomposition into two mutually orthogonal vector fields, one of which is curl-free and the other is divergence-free, then the potential function of the former can serve as a Lyapunov function candidate, since its positive definiteness will reflect the stability of the system. Moreover, under some additional conditions, its sublevel sets will give the exact attraction domain of the system. A sufficient condition for the existence of the proposed vector field decomposition is first obtained for $2$-dimensional systems by solving a partial differential equation and then generalized to $n$-dimensional systems. Furthermore, the proposed vector field decomposition always exists for linear systems and can be obtained by solving a specific algebraic Riccati equation.
\end{abstract}


\end{frontmatter}

\section{Introduction}\label{Sec1}
As is well known, \emph{Lyapunov functions} (LFs) play a key role in the analysis of stability of dynamical systems. Since the Russian mathematician A. M. Lyapunov introduced the famous Lyapunov direct method in his seminal work entitled \emph{The General Problem of Stability of Motion} \cite{Lyapunov1892The}, finding effective methods for constructing LFs has become a central issue of modern control theory. 

The simplest and most commonly used methods may be the \emph{variable gradient method} and \emph{Krasovskii's method}. The variable gradient method assumes a certain form for the gradient of the unknown LF which usually depends on some unknown parameters that need to be determined by the properties of a LF. Although integrating the gradient may produce a LF, there is no unified approach for the selection of the gradient. It relies heavily on experiences and physical intuition. Krasovskii's method is essentially a linearization method, and hence the resulting LF usually can only be used in a neighborhood (which may be very small) of the equilibrium point. Another kind of effective method is the energy-based method, such as \emph{Energy-Casimir method} \cite{Haddad2007Nonlinear} and \emph{Hamiltonian function method} \cite{Wang2003Generalized,Maschke2000Energy}. However, there exist some complex dynamical systems to which none of the aforementioned methods can be applied and one has to seek for numerical methods. Commonly used numerical methods include the radial basis functions based method \cite{Giesl2007Construction}, the linear programming based method \cite{giesl2012construction}, and the iterative based method \cite{argaez2018iterative}. Other kinds of numerical methods can be found in \cite{matsue2017construction,hafstein2007algorithm,anghel2013algorithmic,giesl2019verification} and references therein. Numerical methods are very useful in practice, though, they provide little insight to the structure of dynamical systems. 

Apart from stability, there is another issue that is of great importance in the study of stability of dynamical systems, namely, the \emph{domain of attraction} (DA). Since in practice, it may happen that a dynamical system is stable but with very narrow stability margin, and becomes unstable once disturbed even by slightest external disturbance. However, finding the exact DA of a dynamical system is anything but easy. In fact, the only method noticed by us that does give an accurate characterization of the DA is given by Zubov \cite{zubov1964methods}. But it has to solve a partial differential equation (known as \emph{Zubov's equation}) subjecting to some additional conditions which is difficult in general. Thus, one turns to estimate the DA instead. For example, \emph{Zubov’s Method} \cite{zubov1964methods} estimates the DA of a dynamical system by approximately solving Zubov's equation using the series expansion of the LF. Generalizations of Zubov's method to systems subject to perturbations or control inputs can be found in \cite{camilli2001generalization,camilli2005calculating}. Traditionally, estimates of the DA are given by sublevel sets of LFs. But such estimates are usually too conservative, justifying the term \emph{inner estimates}. Less conservative estimates can be found in \cite{valmorbida2017region} which are given by positively invariant sets that are not necessarily sublevel sets of a LF. Sometimes, however, \emph{outer estimates} containing the DA are more useful in practice such as collision avoidance \cite{korda2014controller}. A method for finding outer estimates of the DA can be found in \cite{henrion2013convex}, where the DA is computed by solving an infinite-dimensional convex linear programming problem over the space of measures. This method is then extended to obtain both inner and outer estimates of the DA \cite{korda2014controller}. Computing both inner and outer estimates, compared to just one or the other, enables assessing the tightness of the estimates obtained and provides a valuable insight into achievable performance and/or safety of a given constrained control system. For a detailed discussion on the estimates of the DA, please refer to the book \cite{chesi2011domain}. 

In this paper, we shall propose a novel vector field decomposition based approach for constructing LFs. This approach stand a chance to be applied to a large class of dynamical systems (depends on the existence of the desired vector field decomposition). Also, it will provide some insights to the structure of the motion of the system. In fact, if the desired vector field decomposition does exist for a given dynamical system, then naturally there is a corresponding LF candidate and the motion of the system can be decomposed into two directions, one within the level surfaces of the LF candidate and the other orthogonal to them. Moreover, under some additional conditions, the DA of the dynamical system is fully characterized by a sublevel set of the LF. We shall also give a sufficient condition for the existence of the proposed vector filed decomposition for general nonlinear dynamical systems. This condition is first obtained for the $2$-dimensional cases by solving a partial differential equation and then generalized to the $n$-dimensional cases by treating the defining vector field as a differential $1$-form. Finally, we shall prove that for linear systems, the proposed vector field decomposition always exists and can be found explicitly by solving a specific algebraic Riccati equation. 

The organization of the paper is as follows. The proposed vector field decomposition based approach is presented in Section \ref{Sec2}. A sufficient condition for the existence of the proposed vector field decomposition for nonlinear systems is given in Section \ref{Sec3}. The existence of the proposed vector field decomposition for linear systems is proved in Section \ref{Sec4}. Finally, the conclusions are given in Section \ref{Sec5}. 

\section{A vector field decomposition based approach}\label{Sec2}

Consider a general nonlinear dynamical system of the form
\begin{equation}\label{NDS}
\dot{x}=f(x),
\end{equation}
where $x\in\mathbb{R}^n$ is the system state vector and $f:\mathbb{R}^n\to\mathbb{R}^n$ is a smooth map (the smoothness can be relaxed but we assume it for simplicity). To see the motivation behind our vector field decomposition based approach, observe that if the dynamical system \eqref{NDS} admits a smooth LF, say $V$, then the map $f$, viewed as a vector field (hereafter, called the defining vector field of the system \eqref{NDS}) on $\mathbb{R}^n$, can be naturally decomposed into the sum of two vector fields as
\begin{equation}\label{VFD}
f(x)=g(x)+h(x),\quad x\in\mathbb{R}^n,
\end{equation}
where $-g(x)=\nabla V(x),x\in\mathbb{R}^n$ is the \emph{gradient field} (and hence \emph{curl-free}) of $V$. Moreover, we have, by the definition of LF, 
\begin{equation}\label{DLF}
\left<\nabla V(x), f(x)\right>=-\left<g(x),g(x)+h(x)\right>\le 0,\quad x\in\mathbb{R}^n,
\end{equation}
where $\left<\cdot,\cdot\right>$ denotes the usual Euclidean inner product. Intuitively, equation \eqref{DLF} means that the sum vector field $f=g+h$ always point `inside' the level surfaces of $V$, i.e., the trajectories of \eqref{NDS} always move from higher level surfaces to lower ones, despite the vector field $h$ may point `inside', `outside', or tangent to the level surfaces. It is noteworthy, however, that in the critical case where the vector field $h$ tangent to the level surfaces of $V$, we have
\begin{equation}\label{OC}
\left<g(x),h(x)\right>=0,\quad x\in\mathbb{R}.
\end{equation}
and \eqref{DLF} always holds since $-\left<g(x),g(x)\right>$ is negative semidefinite.

The geometric meaning of \eqref{OC} is that the vector field $g$ is pointwise orthogonal to $h$ in $\mathbb{R}^n$. Therefore, we would like to decompose the vector field $f$ into a sum of two vector fields as $f=g+h$ such that $g$ is curl-free and orthogonal to $h$ everywhere in $\mathbb{R}^n$, i.e., equation \eqref{OC} holds. Moreover, in order to rule out the trivial solution $g(x)\equiv 0$, we shall also require $h$ to be \emph{divergence-free} (which implies that the vector fields $f$ and $g$ have the same divergence, a property that is not satisfied by the trivial solution $g(x)\equiv 0$ in general).

Suppose the defining vector field $f$ of system \eqref{NDS} has the proposed decomposition, then there exists a unique smooth function $V:\mathbb{R}^n\rightarrow\mathbb{R}$ satisfying $\nabla V(x)=-g(x)$ and $V(0)=0$ such that $V$ is decreasing along the trajectories of system \eqref{NDS}, since 
\begin{equation*}
\frac{\mathrm{d}}{\mathrm{d}t}V(x)=\nabla V(x)f(x)=\left<-g(x),g(x)+h(x)\right>=-\left\Vert g(x)\right\Vert^2\le 0,\quad\forall x\in\mathbb{R}^n.
\end{equation*}
Thus, if $V(x)$ happens to be positive definite, then it is a LF of the system.  Actually, we have

\begin{Theorem}\label{Thm1}
	Consider the dynamical system \eqref{NDS}. Suppose the defining vector field $f$ admits a decomposition $f=g+h$ in a neighborhood $U\subset\mathbb{R}^n$ of $0$ such that $g$ is curl-free, $h$ is divergence-free and $\left<h(x),g(x)\right>=0$ for all $x\in U$. Let $V:U\rightarrow\mathbb{R}$ be the smooth function satisfying $\nabla V(x)=-g(x)$ and $V(0)=0$.
	\begin{enumerate}[1)]
		\item If $V(x)$ is positive definite on $U\backslash\left\{0\right\}$, then the equilibrium point $x=0$ is Lyapunov stable. 
		\item If in addition, $g(x)\ne 0$ for all $x\in U\backslash\{0\}$, then the equilibrium point $x=0$ is asymptotically stable.
		\item Finally, if there exists a scalar $\varepsilon>0$ such that $g(x)\ne0$ for all $x\in\mathcal{B}_\varepsilon(0)\backslash\{0\}$, and for every sufficiently small $\delta>0$ there exists $x_0\in\mathbb{R}^n$ such that $\Vert x_0\Vert<\delta$ and $V(x_0)<0$, then the equilibrium point $x=0$ is unstable.
	\end{enumerate}
\end{Theorem}

\begin{proof}
	1) and 2) follow directly from Lyapunov's stability theorem (see \cite{Haddad2007Nonlinear}, Theorem 3.1), while 3) follows directly from Lyapunov's instability theorem (see \cite{Haddad2007Nonlinear} Theorem 3.12).
\end{proof}

It seems that the proposed scheme not only complicates the searching for a scalar function $V$ satisfying Lyapunov's stability theorem, but also strongly restricts the set of systems and LFs the scheme works for. This issue is clear for the case of linear systems studied in Section \ref{Sec4} where the method reduces to solve a restricted case of the algebraic Lyapunov equation. Nevertheless, it has the virtue that it gives an accurate characterization of the DA. Indeed, we have
\begin{Theorem}\label{Thm2}
	Assume the defining vector field $f$ of \eqref{NDS} admits a decomposition $f=g+h$ as in Theorem \ref{Thm1} and $V:\mathbb{R}^n\rightarrow\mathbb{R}$ be the smooth function satisfying $\nabla V(x)=-g(x)$ and $V(0)=0$. Let $c\in\mathbb{R}$ be positive and $\Omega_c$ the path component of $\{x:V(x)<c\}$ containing the equilibrium point $x=0$. If the closure $\bar{\Omega}_c\subset U$,
	\begin{equation*}
	\begin{split}
	g(x)&=0,\quad x\in\partial\Omega_c,\\
	g(x)&\ne 0,\quad x\in\Omega_c\backslash\{0\},\\
	V(x)&>0,\quad x\in\Omega_c\backslash\{0\},
	\end{split}
	\end{equation*}
	and $V$ is radially unbounded, then $\Omega_c$ is the DA of system \eqref{NDS}. 
\end{Theorem}
\begin{proof}
	First of all, the equilibrium point $x=0$ is asymptotically stable according to Theorem \ref{Thm1}. Let $\mathcal{D}_a$ be the DA of system \eqref{NDS}. Since $V(x)$ is positive definite on $\Omega_c\backslash\{0\}$, there exists $\varepsilon>0$ such that $\Omega_\varepsilon\coloneqq\{x\in\Omega_c:V(x)<\varepsilon\}\subseteq\mathcal{D}_a$. To prove that $\Omega_c=\mathcal{D}_a$, we first show that $\Omega_c\subseteq \mathcal{D}_a$. To this end, it suffices to show that for each $\varepsilon\le c^\prime<c$ and $x_0\in\bar{\Omega}_{c^\prime}\backslash \Omega_\varepsilon$, the trajectory $x(t)$ issued from $x_0$, i.e., that satisfies $x(0)=x_0$, will enter $\Omega_\varepsilon$ in a finite time $t$. By the radially unboundedness of $V$, the set $\bar{\Omega}_{c^\prime}\backslash\Omega_\varepsilon$ is compact. Thus, there exists $\delta>0$ such that $\Vert g(x)\Vert^2>\delta$ for all $x\in\bar{\Omega}_{c^\prime}\backslash\Omega_\varepsilon$, since $g(x)\ne 0$ on this set. Consequently, $x(t)$ will enter $\Omega_\varepsilon$ at $t=(c'-\varepsilon)/\delta$ since we have
	\begin{equation*}
	\frac{\mathrm{d}}{\mathrm{d}t}V(x)=\left<\nabla V(x),f(x)\right>=-\Vert g(x)\Vert^2<-\delta
	\end{equation*}
	and $V(x)\le c^\prime$ in $\bar{\Omega}_{c^\prime}\backslash\Omega_\varepsilon$.
	
	Next, we prove $\mathcal{D}_a\subseteq\Omega_c$. This can be achieved by showing that for each $x_0\in U\backslash\Omega_c$ the trajectory $x(t)$ issued from $x_0$ cannot across the boundary $\partial\Omega_c$ of $\Omega_c$ in any finite time. Suppose on the contrary there does exist a finite $t>0$ such that $x(t)$ crosses the boundary $\partial\Omega_c$ into $\Omega_c$, then the vector field $f$ intersects $\partial\Omega_c$ transversally at $x(t)$. Since $f=g+h$ and $g(x)=0$ at $x(t)$, we conclude that the vector field $h$ intersects $\partial\Omega_c$ transversally at $x(t)$. But then for a sufficiently small $\varepsilon>0$, the vector field $h$ will intersects the lever surface $\{x\in U:V(x)=c^\prime\}$ transversally at $x(t+\varepsilon)$, where $c^\prime=V(x(t+\varepsilon))<c$. This violates the orthogonality of vector fields $g$ and $h$, a contradiction.
\end{proof}

\begin{Example}\label{Exa1}
	Consider the following nonlinear dynamical system adopted from \cite{KHALIL2002Nonlinear},
	\begin{equation*}
	\begin{split}
	\dot{x}_1=-x_1[1+(x_1^2+x_2^2)-(x_1^2+x_2^2)^2]+x_2,\\
	\dot{x}_2=-x_2[1+(x_1^2+x_2^2)-(x_1^2+x_2^2)^2]-x_1.
	\end{split}
	\end{equation*}
	The defining vector field $f$ of it has the desired decomposition $f=g+h$ with 
	\begin{equation*}
	g=-[1+(x_1^2+x_2^2)-(x_1^2+x_2^2)^2](x_1,x_2)^T
	\end{equation*}
	curl-free and 
	$h=(x_2,-x_1)^T$
	divergence-free. Thus, by Theorem \ref{Thm1}, the function
	\begin{equation*}
	V(x)=\frac{1}{2}(x_1^2+x_2^2)\left[1+\frac{1}{2}(x_1^2+x_2^2)-\frac{1}{3}(x_1^2+x_2^2)^2\right] 
	\end{equation*}
	is a LF of the system and the equilibrium point $x=0$ is asymptotically stable. Let $U=\{x\in\mathbb{R}^2:x_1^2+x_2^2<\frac{3+\sqrt{57}}{4}\}$ and $\Omega=\{x\in\mathbb{R}^2:x_1^2+x_2^2<\frac{1+\sqrt{5}}{2}\}$. Then $V$ is positive definite on $U\backslash\{0\}$ and $g(x)=0$ for $x\in\partial\Omega$. Moreover, $\Omega$ is the path component of a sublevel set of $V$ that contains $x=0$, and the conditions of Theorem \ref{Thm2} are satisfied. Therefore, $\Omega$ is the DA of the system. This can be seen via the phase portrait of the system.
\end{Example}

At this point, a natural question is whether the decomposition required by Theorem \ref{Thm1} and \ref{Thm2} exists. It is a well known result in physics (especially in electromagnetics) that one can always decompose a vector field into the sum of a divergence-free and a curl-free vector field, and such decomposition is not unique. But with the additional orthogonality condition \eqref{OC}, the question turns out to be much harder. Nevertheless, we shall give a sufficient condition for the existence of the proposed decomposition in the next section. Furthermore, a complete answer for linear systems will be given in Section \ref{Sec4}.

\section{A sufficient condition for nonlinear systems}\label{Sec3}

In this section, we give a sufficient condition for the existence of the proposed vector field decomposition for general nonlinear systems. To this end, we first address the simpler $2$-dimensional cases and then generalize the result to $n$-dimensional cases.

\subsection{$2$-dimensional cases}
Given a smooth vector field $f$ on $\mathbb{R}^2$, suppose it can be decomposed into the sum of two smooth vector fields as $f=g+h$, where $g$ is curl-free, $h$ is divergence-free, and $\left<g(x),h(x)\right>=0$ for all $x\in\mathbb{R}^2$. Then, by the irrotationality of $g$, i.e.,
\begin{equation*}
\frac{\partial g_1}{\partial x^2}=\frac{\partial g_2}{\partial x^1},
\end{equation*}
there exists (see the Poincar\'e Lemma in \cite{Lee2012Introduction}) a scaler function $\Phi:\mathbb{R}^2\to\mathbb{R}$ such that 
\begin{equation}\label{g}
g_1=\frac{\partial\Phi}{\partial x^1}\quad\text{and}\quad g_2=\frac{\partial\Phi}{\partial x^2},
\end{equation} 
where $g_i$ and $x^i,i=1,2$ denote the $i$-th component of $g$ and $x$, respectively. Similarly, there exists a scalar function $\Psi:\mathbb{R}^2\rightarrow\mathbb{R}$ such that
\begin{equation}\label{h}
h_1=-\frac{\partial \Psi}{\partial x^2}\quad\text{and}\quad h_2=\frac{\partial \Psi}{\partial x^1},
\end{equation}
since $h$ is divergence-free. Plugging \eqref{g} and \eqref{h} into $\left<g,h\right>=0$ we find that 
\begin{equation}\label{PDE}
-\frac{\partial\Phi}{\partial x^1}\frac{\partial\Psi}{\partial x^2}+\frac{\partial\Phi}{\partial x^2}\frac{\partial\Psi}{\partial x^1}=0,\quad x\in\mathbb{R}^2.
\end{equation}

If the vector field $g$ is nondegenerate on an open set $U\subset\mathbb{R}^2$, then considered as a partial differential equation of $\Psi$, the classical differential equation theory tells us that the general solution of equation \eqref{PDE} can be found by integrating the corresponding characteristic system of ordinary differential equation (see \cite{Grigoriev2010Symmetries,Olver1986Applications} for details), i.e., 
\begin{equation}\label{ODE}
\frac{dx^1}{{\partial\Phi}/{\partial x^2}}=\frac{dx^2}{-{\partial\Phi}/{\partial x^1}},\quad x\in U.
\end{equation}
Equation \eqref{ODE} can be solved explicitly and the solution is 
\begin{equation*}
\Phi(x)=c,\quad x\in U,
\end{equation*}
where $c$ is the constant of integration. And any solution of \eqref{PDE} can be written as a function of $c$, i.e., $\Psi(x)=\alpha(c)=\alpha(\Phi(x))$ for a smooth function $\alpha:\mathbb{R}\rightarrow\mathbb{R}$. 
Therefore, we have
\begin{equation}\label{SC1}
\begin{split}
f_1&=\frac{\partial\Phi}{\partial x^1}-\alpha'(\Phi)\frac{\partial\Phi}{\partial x^2},\\
f_2&=\frac{\partial\Phi}{\partial x^2}+\alpha'(\Phi)\frac{\partial\Phi}{\partial x^1},\\
\end{split}
\end{equation} 
where $\alpha'$ denotes the derivative of $\alpha$ with respect to $\Phi$. Clearly, if a vector field $f$ can be written in the form of \eqref{SC1}, then it has the desired decomposition. 

On the other hand, if $h$ is nondegenerate on an open set $U$, then by the same argument we will find that $\Phi=\beta(\Psi)$ on $U$ for a smooth function $\beta:\mathbb{R}\to\mathbb{R}$ and consequently, we have
\begin{equation}\label{SC2}
\begin{split}
f_1&=\beta'(\Psi)\frac{\partial\Psi}{\partial x^1}-\frac{\partial\Psi}{\partial x^2},\\
f_2&=\beta'(\Psi)\frac{\partial\Psi}{\partial x^2}+\frac{\partial\Psi}{\partial x^1}.
\end{split}
\end{equation}
Combining equations \eqref{SC1} and \eqref{SC2} yields the following sufficient condition for the existence of the proposed vector field decomposition.

\begin{Theorem}\label{Thm3}
	Let $f=(f_1,f_2)^T$ be a smooth vector field on $\mathbb{R}^2$. If there exist smooth functions $\Theta:\mathbb{R}^2\rightarrow\mathbb{R}$ and $\alpha,\beta:\mathbb{R}\rightarrow\mathbb{R}$ such that
	\begin{equation}\label{SC}
	\begin{split}
	f_1&=\beta(\Theta)\frac{\partial\Theta}{\partial x^1}-\alpha(\Theta)\frac{\partial\Theta}{\partial x^2},\\
	f_2&=\beta(\Theta)\frac{\partial\Theta}{\partial x^2}+\alpha(\Theta)\frac{\partial\Theta}{\partial x^1}
	\end{split}
	\end{equation}
	then $f$ admits the proposed decomposition $f=g+h$ with $g=\beta(\Theta)(\frac{\partial\Theta}{\partial x^1},\frac{\partial\Theta}{\partial x^2})^T$ and $h=\alpha(\Theta)(-\frac{\partial\Theta}{\partial x^2},\frac{\partial\Theta}{\partial x^1})^T$.  
\end{Theorem}

For instance, the defining vector field of the nonlinear system given in Example \ref{Exa2} satisfies the conditions of Theorem \ref{Thm3}, and hence admits the proposed decomposition as we have seen before.

\begin{Remark}
	It is clear that any smooth vector field $f$ on $\mathbb{R}^2$ can be written in the form
	\begin{equation}\label{ANY}
	\begin{split}
	f_1&=\beta(x)\frac{\partial\Theta}{\partial x^1}-\alpha(x)\frac{\partial\Theta}{\partial x^2},\\
	f_2&=\beta(x)\frac{\partial\Theta}{\partial x^2}+\alpha(x)\frac{\partial\Theta}{\partial x^1},\\
	\end{split}
	\end{equation}
	for some smooth functions $\Theta:\mathbb{R}^2\to\mathbb{R}$ and $\alpha,\beta:\mathbb{R}^2\to\mathbb{R}$. What Theorem \ref{Thm3} tells us is that if the function $\Theta$ can be chosen in a manner such that $\alpha$ and $\beta$ are constant on each level set of $\Theta$, then the existence of the proposed decomposition is guaranteed. The question is, under this restriction, whether equation \eqref{ANY} is still sufficient to represent all vector fields on $\mathbb{R}^2$.
\end{Remark}

\subsection{$n$-dimensional cases}
In order to generalize Theorem \ref{Thm3} to $n$-dimensional cases, we have to introduce some more notations. For high dimensional cases, it is more convenient to treat a vector filed $f$ on $\mathbb{R}^n$ as a \emph{differential $1$-form}, i.e.,
\begin{equation*}
f=\sum_{i=1}^nf_i\mathrm{d}x^i.
\end{equation*}
Under this new notation, the curl of $f$ can be written as 
\begin{equation}\label{curl}
\mathrm{curl}f=\mathrm{d}f=\sum_{i=1}^n\sum_{j=1}^n\frac{\partial f_i}{\partial x^j}\mathrm{d}x^j\wedge\mathrm{d}x^i=\sum_{1\le i<j\le n}\left(\frac{\partial f_i}{\partial x^j}-\frac{\partial f_j}{\partial x^i}\right)\mathrm{d}x^j\wedge\mathrm{d}x^i,
\end{equation}
where $\mathrm{d}$ is the \emph{exterior derivative operator}. If $f$ is curl-free, i.e., $\mathrm{d}f=0$, then there exist a $0$-form (i.e., a scaler function) $\Phi$ such that $f=\mathrm{d}\Phi$. 

In addition, let $\mathrm{d}^\ast$ be the $L^2$-adjoint of the exterior derivative $\mathrm{d}$, then the divergence of $f$ can be written as
\begin{equation}\label{div}
\begin{split}
\mathrm{div}f&=-\mathrm{d}^\ast f=\ast\mathrm{d}\ast f\\
&=\ast\mathrm{d}\sum_{i=1}^n(-1)^{i-1}f_i\mathrm{d}x^1\wedge\cdots\wedge\widehat{\mathrm{d}x^i}\wedge\cdots\wedge\mathrm{d}x^n\\
&=\ast\sum_{i=1}^n\frac{\partial f_i}{\partial x^i}\mathrm{d}x^1\wedge\cdots\wedge\mathrm{d}x^n\\
&=\sum_{i=1}^n\frac{\partial f_i}{\partial x^i},
\end{split}
\end{equation} 
where $\ast$ is the \emph{Hodge star operator} (the precise definition of $\mathrm{d}^\ast$ and $\ast$ as well as their relationship can be found in \cite{Jost2011Riemannian}). If $f$ is divergence free, i.e., $\mathrm{d}^\ast f=0$, then there exists a $2$-form 
\begin{equation*}
\psi = \sum_{1\le i<j\le n}\Psi_{ij}\mathrm{d}x^i\wedge\mathrm{d}x^j
\end{equation*}
such that $f=\mathrm{d}^\ast \psi$. Now we claim:

\begin{Theorem}\label{Thm4}
	Let $f=\sum_{i=1}^nf_i\mathrm{d}x^i$ be a vector field on $\mathbb{R}^n$, if there exist smooth functions $\Theta:\mathbb{R}^n\rightarrow\mathbb{R}$, $\alpha_{ij}:\mathbb{R}\rightarrow\mathbb{R}$, and $\beta:\mathbb{R}\to\mathbb{R}$ such that
	\begin{equation*}
	f = \mathrm{d}\Phi+\mathrm{d}^\ast\psi
	\end{equation*}
	hold with
	\begin{equation}\label{Phi_Psi}
	\Phi=\beta(\Theta)\quad\text{and}\quad\Psi_{ij}=\alpha_{ij}(\Theta),
	\end{equation}
	then $f$ admits the proposed decomposition $f=g+h$ with $g=\mathrm{d}\Phi$ and $h=\mathrm{d}^\ast\psi$.
\end{Theorem}
\begin{proof}
	Apparently, $g$ is divergence-free and $f$ is curl-free, since $\mathrm{d}\circ\mathrm{d}=0$ and $\mathrm{d}^\ast\circ\mathrm{d}^\ast=0$. Moreover, we have 
	\begin{equation*}
	g=\mathrm{d}\Phi=\sum_{i=1}^n\frac{\partial\Phi}{\partial x^i}\mathrm{d}x^i
	\end{equation*}
	and
	\begin{equation*}
	\begin{split}
	h=&\mathrm{d}^\ast\psi=(-1)^{n(2+1)+1}\ast\mathrm{d}\ast\sum_{1\le i<j\le n}\Psi_{ij}\mathrm{d}x^i\wedge\mathrm{d}x^j\\
	=&(-1)^{n+1}\ast\mathrm{d}\sum_{1\le i<j\le n}(-1)^{i+j-1}\Psi_{ij}\mathrm{d}x^1\wedge\cdots\wedge\widehat{\mathrm{d}x^i}\wedge\cdots\wedge\widehat{\mathrm{d}x^j}\wedge\cdots\wedge\mathrm{d}x^n\\
	=&(-1)^{n+1}\ast\sum_{1\le i<j\le n}(-1)^j\frac{\partial\Psi_{ij}}{\partial x^i}\mathrm{d}x^1\wedge\cdots\wedge\widehat{\mathrm{d}x^j}\wedge\cdots\wedge\mathrm{d}x^n\\
	&+(-1)^{n+1}\ast\sum_{1\le i<j\le n}(-1)^{i-1}\frac{\partial\Psi_{ij}}{\partial x^j}\mathrm{d}x^1\wedge\cdots\wedge\widehat{\mathrm{d}x^i}\wedge\cdots\wedge\mathrm{d}x^n\\
	=&(-1)^{n+1}\sum_{1\le i<j\le n}(-1)^{n}\frac{\partial\Psi_{ij}}{\partial x^i}\mathrm{d}x^j-\sum_{1\le i<j\le n}(-1)^{n}\frac{\partial\Psi_{ij}}{\partial x^j}\mathrm{d}x^i\\
	=&\sum_{1\le i<j\le n}\left(\frac{\partial\Psi_{ij}}{\partial x^j}\mathrm{d}x^i-\frac{\partial\Psi_{ij}}{\partial x^i}\mathrm{d}x^j\right).
	\end{split}
	\end{equation*} 
	Therefore,
	\begin{equation*}
    \begin{split}
	\left<g,h\right>&=\left<\sum_{k=1}^n\frac{\partial\Phi}{\partial x^k}\mathrm{d}x^k,\sum_{1\le i<j\le n}\left(\frac{\partial\Psi_{ij}}{\partial x^j}\mathrm{d}x^i-\frac{\partial\Psi_{ij}}{\partial x^i}\mathrm{d}x^j\right)\right>\\
	&=\sum_{k=1}^n\sum_{1\le i<j\le n}\frac{\partial\Phi}{\partial x^k}\left(\frac{\partial\Psi_{ij}}{\partial x^j}\delta^{ki}-\frac{\partial\Psi_{ij}}{\partial x^i}\delta^{kj}\right)\\
	&=\sum_{1\le i<j\le n}\left(\frac{\partial\Psi_{ij}}{\partial x^j}\frac{\partial\Phi}{\partial x^i}-\frac{\partial\Psi_{ij}}{\partial x^i}\frac{\partial\Phi}{\partial x^j}\right)\\
	&=\sum_{1\le i<j\le n}\frac{\mathrm{d}\alpha_{ij}}{\mathrm{d}\Theta}\frac{\mathrm{d}\beta}{\mathrm{d}\Theta}\left(\frac{\partial\Theta}{\partial x^j}\frac{\partial\Theta}{\partial x^i}-\frac{\partial\Theta}{\partial x^i}\frac{\partial\Theta}{\partial x^j}\right)\\
	&=0,
    \end{split}
	\end{equation*}
	wehre in the second last equation we have used the condition \eqref{Phi_Psi}. Thus, $f=g+h$ is the desired vector field decomposition.
\end{proof}

\begin{Remark}\label{Rem2}
	The remarkable \emph{Hodge decomposition theorem} (see Corollary 3.4.2 in \cite{Jost2011Riemannian}) says that any vector field can be decomposed into the sum of a curl-free, a divergence-free, and a harmonic vector field, such that they are orthogonal to each other with respect to the $L^2$ inner product. One can hardly overlook the intimate connection between the proposed decomposition and the Hodge decomposition.
\end{Remark}

\section{Existence for linear systems}\label{Sec4}

To prove the existence for linear systems, we shall first recall some basics about algebraic Riccati equation. The contents are mainly based on the book \cite{Gohberg2005Indefinite} of I. Gohberg, P. Lancaster and L. Rodman. A general \emph{continuous algebraic Riccati equation} (CARE) is a quadratic matrix equation of the form
\begin{equation}\label{CARE}
XDX+XA+BX-C=0,
\end{equation}
where $D,A,B,C$ are real or complex matrices of sizes $n\times m,n\times n,m\times m$, and $m\times n$, respectively. Note that the coefficients of CARE \eqref{CARE} fit into a square matrix of size $(n+m)\times(n+m)$:
\begin{equation*}
T=\begin{pmatrix}
A&D\\
C&-B\\
\end{pmatrix}.
\end{equation*}
The solutions of \eqref{CARE} can be fully characterized by the Jordan chains of $T$.
\begin{Definition}[\cite{Gohberg2005Indefinite}]\label{Def3}
	Let $T$ be an $n\times n$ matrix and $\lambda$ an eigenvalue of $T$, if there exist $\eta_i\in\mathbb{C}^n$ satisfying $\eta_i\ne 0,i=1,\cdots,l$, such that
	\begin{equation*}
	T\eta_1=\lambda\eta_1,\ \ \text{and  }\ \ T\eta_i=\lambda\eta_i+\eta_{i-1},\quad i=2,\cdots,l,
	\end{equation*}
	then the ordered $l$-tuple $(\eta_1,\cdots,\eta_l)$ is called a \emph{Jordan chain} of $T$ corresponding to the eigenvalue $\lambda$. 
\end{Definition}

The following constructive theorem then gives the solutions of CARE \eqref{CARE}. 

\begin{Theorem}[\cite{Gohberg2005Indefinite}]\label{Thm5}
	Equation \eqref{CARE} has a solution $X\in\mathbb{C}^{m\times n}$ if and only if there exist vectors $v_1,\cdots,v_n\in\mathbb{C}^{m+n}$ that form a set of Jordan chains for $T$ and, if
	\begin{equation*}
	v_j=\begin{pmatrix}
	y_j\\
	z_j\\
	\end{pmatrix},\quad j=1,2,\cdots,
	\end{equation*}
	where $y_j\in\mathbb{C}^n$, then $y_1,\cdots,y_n$ form a basis for $\mathbb{C}^n$. Furthermore, if 
	\begin{equation*}
	Y=(y_1,\cdots,y_n),\quad Z=(z_1,\cdots,z_n),
	\end{equation*}
	then every solution of \eqref{CARE} has the form $X=ZY^{-1}$ for some set of Jordan chains $v_1,\cdots,v_n$ for $T$ such that $Y$ is nonsingular.
\end{Theorem}

Now we can prove the existence of the proposed vector field decomposition for linear systems. Note that for linear systems we have $f(x)=Fx$ for a certain $n\times n$ matrix $F\in\mathbb{R}^{n\times n}$. Let $g(x)=Gx$ for some $G\in\mathbb{R}^{n\times n}$, then by the irrotationality of $g$ and equation \eqref{curl}, we find that $G_{ij}=G_{ji}$, i.e., $G$ is a symmetric matrix. On the other hand, according to \eqref{div} we find that
\begin{equation}\label{Eq4.1}
\mathrm{tr}G=\mathrm{tr}F,
\end{equation}
since $h=f-g$ is divergence free. Moreover, the orthogonality condition \eqref{OC} becomes to $x^TG^T(F-G)x=0$ which implies that $G^T(F-G)$ is a skew-symmetric matrix, that is $G^T(F-G)=-(F-G)^TG$, or equivalently
\begin{equation}\label{SARE}
GF+F^TG=2G^2,
\end{equation}
where we have utilized the symmetry of matrix $G$. So we are going to find a symmetric solution of \eqref{SARE} subjecting to constraint \eqref{Eq4.1}.

Note that equation \eqref{SARE} is a special form of the algebraic Riccati equation \eqref{CARE}. Indeed, if we take $m=n$, $D=2I,X=G,A=-F,B=-F^T$, and $C=0$, then equation \eqref{CARE} becomes to equation \eqref{SARE}. Therefore, all the solutions of \eqref{SARE} can be found by Theorem \ref{Thm5}. Among these solutions, we have to prove that there exists a symmetric one such that constraint \eqref{Eq4.1} is fulfilled. To begin with, let us look at a simple illustrative example.

\begin{Example}\label{Exa2}
	Let $D=2I_{2\times 2},C=0_{2\times 2}, A=\begin{pmatrix}
	1&1\\
	0&0\\
	\end{pmatrix}$, and $B=A^T$, then $T=PJP^{-1}$ with 
	\begin{equation*}
	J=\begin{pmatrix}
	0&1&0&0\\
	0&0&0&0\\
	0&0&-1&0\\
	0&0&0&1\\
	\end{pmatrix},\quad\text{and}\quad
	P=\begin{pmatrix}
	2&-2&2&0\\
	-2&0&0&2\\
	0&0&0&1\\
	0&-1&0&1\\
	\end{pmatrix},
	\end{equation*}
	be the Jordan normal form and the corresponding transformation matrix of $T$, respectively. There are $4$ sets of Jordan chains for $T$ satisfying the conditions of Theorem \ref{Thm5}, i.e., $(P_1,P_2),(P_1,P_3),(P_1,P_4),(P_3,P_4)$, where $P_i$ denotes the $i$-th column of $P$. Hence, there are four solutions for equation \eqref{CARE}, i.e.,
	\begin{equation*}
	X_1=\begin{pmatrix}
	0&0\\
	\frac{1}{2}&\frac{1}{2}\\
	\end{pmatrix},\quad X_2=\begin{pmatrix}
	0&0\\
	0&0\\
	\end{pmatrix},\quad X_3=\begin{pmatrix}
	\frac{1}{2}&\frac{1}{2}\\
	\frac{1}{2}&\frac{1}{2}\\
	\end{pmatrix},\quad\text{and}\quad X_4=\begin{pmatrix}
	0&\frac{1}{2}\\
	0&\frac{1}{2}\\
	\end{pmatrix}.
	\end{equation*}
	Observe that solution $X_3$ is the one with the same trace as $A$, and it corresponds to the eigenvalues $0$ and $1$ of $T$ which are exactly the eigenvalues of $A$. 
\end{Example}

Indeed, the observation is justified by the following lemma.
\begin{Lemma}\label{Thm6}
	Let $X$ be a solution of equation \eqref{CARE} and $v_1,\cdots,v_n\in\mathbb{C}^{m+n}$ the corresponding set of Jordan chains for $T$ satisfying the conditions of Theorem \ref{Thm5}. If $\lambda_1,\cdots,\lambda_n$ are the eigenvalues of $T$ corresponding to $v_1,\cdots,v_n$, then we have
	\begin{equation*}
	\mathrm{tr}\left(A+DX\right)=\sum_{i=1}^n{\lambda_i}.
	\end{equation*}
\end{Lemma} 
\begin{proof}
	Since $X$ is a solution of \eqref{CARE}, we have, according to Proposition 14.4.1 in \cite{Gohberg2005Indefinite}, that the equation 
	\begin{equation*}
	T\begin{pmatrix}
	I\\
	X\\
	\end{pmatrix}=\begin{pmatrix}
	I\\
	X\\
	\end{pmatrix}Q
	\end{equation*}
	holds with $Q=A+DX$. Let $J$ be the Jordan normal form of $T$ and $P$ the corresponding transformation matrix, i.e. $T=PJP^{-1}$. Then, multiplying both sides of the above equation by $P^{-1}$ on the left yields 
	\begin{equation*}
	JP^{-1}\begin{pmatrix}
	I\\
	X\\
	\end{pmatrix}=P^{-1}\begin{pmatrix}
	I\\
	X\\
	\end{pmatrix}Q.
	\end{equation*}
	From Theorem \ref{Thm5}, we know that 
	\begin{equation*}
	\begin{pmatrix}
	I\\
	X\\
	\end{pmatrix}=vY^{-1},
	\end{equation*} 
	where $v=\left(v_1,\cdots,v_n\right)$ and $Y$ is the same as defined in Theorem \ref{Thm5}. Therefore,
	\begin{equation}\label{JP}
	JP^{-1}v=P^{-1}v\tilde{Q}
	\end{equation}
	with $\tilde{Q}=Y^{-1}QY$. Since $v_1,\cdots,v_n$ form a set of Jordan chains for $T$, there exist $1\le i_1,\cdots,i_n\le m+n$ such that $v_j=P_{i_j},1\le j\le n$, where $P_{i_j}$ is the $i_j$-th column of matrix $P$, and $i_j\ne i_k$ whenever $j\ne k$. Then we have $P^{-1}v=\left(e_{i_1},\cdots,e_{i_n}\right)$, where $e_1,\cdots,e_{m+n}$ are the standard bases of $\mathbb{R}^{m+n}$. Consequently, equation \eqref{JP} is equivalent to $\tilde{J}=\tilde{Q}$, where $\tilde{J}=\left(J_{i_ki_l}\right),1\le k,l,\le n$ is the matrix consists of entries of $J$ at $i_k$-th row and $i_l$-th column. It is clear that $\tilde{J}$ is itself in Jordan normal form since $v_1,\cdots,v_n$ form a set of Jordan chains of $T$. Moreover, the diagonal entries of $\tilde{J}$ are exactly the eigenvalues $\lambda_1,\cdots,\lambda_n$, hence $\mathrm{tr}\left(A+DX\right)=\mathrm{tr}\tilde{Q}=\mathrm{tr}\tilde{J}=\sum_{i=1}^n{\lambda_i}$, as expected.
\end{proof}

\begin{Remark}\label{Rem3}
	In the special case of \eqref{SARE}, the eigenvalues of $F$ are also eigenvalues of $T$. If corresponding to them there exists a solution $G$, then the lemma says that $\mathrm{tr}(-F+2G)=\mathrm{tr}F$ which is exactly \eqref{Eq4.1}.
\end{Remark}

Now, we prove the existence of the proposed vector field decomposition for a special class of linear systems where $F$ satisfies some additional conditions. 
\begin{Theorem}\label{Thm7}
	Let $F\in\mathbb{R}^{n\times n}$ be a real matrix and $\lambda_1,\cdots,\lambda_n$ the eigenvalues of it. If $\lambda_i+\lambda_j\ne0$ for any $1\le i,j\le n$, then equation \eqref{SARE} always has a real symmetric solution $G$ such that $\mathrm{tr}G=\mathrm{tr}F$.
\end{Theorem}
\begin{proof}	
	Let $-F=PJP^{-1}$ and $F^T=UKU^{-1}$, where $J,K$ are Jordan normal forms and $P,U$ the corresponding transformation matrices of $-F$ and $F^T$, respectively. Since $\lambda_i+\lambda_j\ne 0$ for $1\le i,j\le n$, $-F$ and $F^T$ have no eigenvalue in common. Thus, there is an invertible matrix $W$ such that $T=WLW^{-1}$ with
	\begin{equation*}
	T=\begin{pmatrix}
	-F&2I\\
	0&F^T\\
	\end{pmatrix}
	\quad\text{and}\quad L = \begin{pmatrix}
	K&0\\
	0&J
	\end{pmatrix}
	\end{equation*}
	a Jordan normal form of $T$. Indeed, we can choose  
	\begin{equation*}
	W=\begin{pmatrix}
	W_{11}&W_{12}\\
	W_{21}&W_{22}
	\end{pmatrix}
	\end{equation*}
	with $W_{11}=XU,W_{12}=P,W_{21}=U,W_{22}=0$, and $X$ the unique (See \cite{Luenberger1965Invertible,Shestopal1976Solution} for the existence and uniqueness.) solution of the \emph{Sylvester equation}
	\begin{equation}\label{SE}
	XF^T+FX=2I.
	\end{equation} 
	
	We claim that the solution $X$ is invertible. Indeed, $X$ is symmetric since apparently $X^T$ is also a solution of \eqref{SE}. Thus, \eqref{SE} becomes to 
	\begin{equation*}
	(FX)^T+FX=2I,
	\end{equation*}
	from which we conclude that $FX$ and consequently $X$ is invertible. But then $W_{11}$ is also invertible and the condition of Theorem \ref{Thm5} is satisfied with $Y=W_{11}$ and $Z=W_{21}=U$. Therefore $G=ZY^{-1}=X^{-1}$ is a solution of \eqref{SARE}. Note that $\mathrm{tr}G=\mathrm{tr}F$ (See Remark \ref{Rem3}) and $G$ is symmetric since $X$ is. And the reality of $G$ is guaranteed by Theorem 14.10.2 in \cite{Gohberg2005Indefinite}. So the proof is complete.
\end{proof} 

Based on this result, we show that the proposed vector field decomposition exists for arbitrary linear system.
\begin{Theorem}\label{Thm8}
	For any real matrix $F\in\mathbb{R}^{n\times n}$, equation \eqref{SARE} always has a real symmetric solution $G$ such that $\mathrm{tr}F=\mathrm{tr}G$.
\end{Theorem}

\begin{proof}
	Suppose \eqref{SARE} has a real symmetric solution $G$ such that $\mathrm{tr}F=\mathrm{tr}G$. Then there exists an orthogonal matrix $Q$ such that $G=Q\Lambda Q^T$, where $\Lambda$ is a diagonal matrix such that $\mathrm{tr}\Lambda=\mathrm{tr}F$. Plugging this into \eqref{SARE} and multiplying both sides on the left by $Q^T$ and right by $Q$ yield
	\begin{equation*}
	\Lambda Q^TFQ+Q^TF^TQ\Lambda =2\Lambda^2,
	\end{equation*}
	which implies that 
	\begin{equation}\label{QFQ}
	\Lambda Q^TFQ=\Lambda^2+B
	\end{equation}
	for a certain skew-symmetric matrix $B$.
	
	Conversely, if there exist an orthogonal matrix $Q$, a skew-symmetric matrix $B$, and a diagonal matrix $\Lambda$ satisfying $\mathrm{tr}\Lambda=\mathrm{tr}F$ such that equation \eqref{QFQ} holds, then \eqref{SARE} is satisfied with $G=Q\Lambda Q^T$. Thus, to prove the theorem, it suffices to show that for any matrix $F\in\mathbb{R}^{n\times n}$, such matrices $Q,B$ and $\Lambda$ do exist. To this end, we assume the set of eigenvalues of $F$ is $\sigma=\{\lambda_1,\cdots,\lambda_n\}$ and construct two subsets $\sigma_1$ and $\sigma_2$ of it iteratively as follows.
	\begin{enumerate}[step 1.]
		\item Let $\sigma_1=\sigma$ and $\sigma_2=\emptyset$.
		\item If we can find a pair of eigenvalues $\lambda_i$ and $\lambda_j$ ($i,j$ may be equal) in $\sigma_1$ such that $\lambda_i+\lambda_j=0$, then go to step \ref{step3}. Else, go to step \ref{step4}.\label{step2}
		\item Remove $\lambda_i$ and $\lambda_j$ from $\sigma_1$ and put them into $\sigma_2$, return to Step \ref{step2}.\label{step3}
		\item End!\label{step4}
	\end{enumerate}
	
	Apparently, the process ends after a finite number of steps and yields two sets $\sigma_1,\sigma_2$ with the following properties:
	\begin{enumerate}[1)]
		\item If $\lambda\in\sigma_1$ (or $\in\sigma_2$), then $\bar{\lambda}\in\sigma_1$ (or $\in\sigma_2$, respectively) too, where $\bar{\lambda}$ denotes the complex conjugate of $\lambda$.
		\item $\lambda_i+\lambda_j\ne 0$ ($i,j$ may be equal) for any two eigenvalues $\lambda_i,\lambda_j\in\sigma_1$.
		\item $\sum_{\lambda\in\sigma_2}{\lambda}=0$.
		\item $\sigma=\sigma_1\cup\sigma_2$.
	\end{enumerate}
	
	According to the well-known real Schur decomposition theorem (See Theorem 7.4.1 in \cite{Golub2013Matrix}), there exists an orthogonal matrix $U\in\mathbb{R}^{n\times n}$ such that   
	\begin{equation*}
	U^TFU=\begin{pmatrix}
	R_{11}&0&\cdots&0\\
	R_{21}&R_{22}&\cdots&0\\
	\vdots&\vdots&\ddots&\vdots\\
	R_{m1}&R_{m2}&\cdots&R_{mm}\\
	\end{pmatrix}\eqqcolon R\in\mathbb{R}^{n\times n},
	\end{equation*}
	where each $R_{ii}$ is a $1\times 1$ matrix or a $2\times 2$ matrix having complex conjugate eigenvalues. In view of properties 1)-4) of $\sigma_1$ and $\sigma_2$, we can, without loss of generality, always assume that $R$ can be divided into blocks as 
	\begin{equation*}
	R=\begin{pmatrix}
	\mathcal{R}_{11}&0\\
	\mathcal{R}_{21}&\mathcal{R}_{22}\\
	\end{pmatrix},
	\end{equation*}
	where $\mathcal{R}_{11}\in\mathbb{R}^{p\times p}$ is a matrix with $\sigma_1$ as its set of eigenvalues (which implies that $\mathrm{tr} \mathcal{R}_{11}=\mathrm{tr} F$ in view of property 3)) and $\mathcal{R}_{22}\in\mathbb{R}^{(n-p)\times(n-p)}$ a matrix with $\sigma_2$ as its set of eigenvalues. Since $\mathcal{R}_{11}$ satisfies the condition of Theorem \ref{Thm7}, there exist, according to the discussion above, an orthogonal matrix $\tilde{Q}$, a diagonal matrix $\tilde{\Lambda}$, and a skew-matrix $\tilde{B}$ such that $\tilde{\Lambda}\tilde{Q}^T\mathcal{R}_{11}\tilde{Q}=\tilde{\Lambda}^2+\tilde{B}$. Then, a simple calculation shows that \eqref{QFQ} is satisfied with
	\begin{equation*}
	\Lambda=\begin{pmatrix}
	\tilde{\Lambda}&0\\
	0&0\\\end{pmatrix},\quad B=\begin{pmatrix}
	\tilde{B}&0\\
	0&0\\
	\end{pmatrix},\quad\text{and}\quad Q=U\begin{pmatrix}
	\tilde{Q}&0\\
	0&I\\
	\end{pmatrix},
	\end{equation*} 
	which completes the proof.
\end{proof}

\section{Conclusions}\label{Sec5}
In this paper, a novel approach for constructing Lyapunov functions was proposed. The advantage of this approach is that under some additional conditions, the resulting Lyapunov function can give an accurate characterization of the attraction domain of the given dynamical system. The approach make use of a specific decomposition of the defining vector field of the system. For nonlinear systems, a sufficient condition for the existence of such decompositions was given. For linear systems, it was proved that such decompositions always exist and can be found explicitly by solving a specific algebraic Riccati equation.

\bibliography{ref}

\end{document}